\documentclass{article}

\usepackage{amsmath,amssymb,amsthm}
\usepackage{mathtools}
\usepackage{bm}
\usepackage{tikz}
\usetikzlibrary{calc}
\usepackage[colorlinks=true,allcolors=blue]{hyperref}
\usepackage{enumitem}
\usepackage{natbib}
\newtheorem{theorem}{Theorem}
\newtheorem{proposition}{Proposition}
\newtheorem{lemma}{Lemma}
\newtheorem{corollary}{Corollary}
\theoremstyle{definition}
\newtheorem{definition}{Definition}
\newtheorem{assumption}{Assumption}
\theoremstyle{remark}
\newtheorem{remark}{Remark}

\newcommand{\Prob}{\mathbb{P}}
\newcommand{\simplex}{\Delta}                  \newcommand{\bidset}{\mathcal{B}}               \newcommand{\BR}{\operatorname{BR}}      \newcommand{\payoff}{\pi}                         \newcommand{\policy}{\gamma}             \newcommand{\dirac}[1]{\delta_{#1}}                                                     \newcommand{\indic}[1]{\mathbf{1}\!\left\{#1\right\}}

\usepackage{xcolor}
\usepackage{amsmath}
\DeclareMathOperator*{\argmax}{arg\,max}

\usepackage{xcolor}
\usepackage{tcolorbox}

\title{On the Role of Tie-Breaking Rules in the Convergence of Fictitious Play for Symmetric First-Price Auctions}

\author{Benjamin Heymann\footnote{Criteo AI Lab, FairPlay joint team, Paris, France}}

\begin{document}

\maketitle
\begin{abstract}
We study continuous-time fictitious play in 2-bidder, symmetric first-price
auctions with independently distributed discrete values and a discrete
bid set. We first exhibit a minimal instance --- two bidders, two
values, three positive bids --- on which fictitious play with the
standard uniform-split tie-breaking rule does \emph{not} converge to
the symmetric Bayes--Nash equilibrium: the equilibrium is 
unstable and the dynamics converge to a stable limit cycle far from the Nash equilibrium. We then show that a
small modification of the tie-breaking rule --- awarding a payoff of
zero to every bidder in case of a tie --- restores convergence:
fictitious play converges to a Nash equilibrium of the modified game.
This limit is an $\epsilon$-equilibrium of the original
auction in a broad range of settings. 
\end{abstract}

\section{Introduction}
\label{sec:intro}

Fictitious play is a classical game-theoretic learning dynamic where players update beliefs about opponents by tracking the empirical frequencies of their actions. At each iteration, every player computes and plays a best response to these beliefs.
Fictitious play was introduced by \citet{brown1951iterative} and
\citet{robinson1951iterative} the same year independently, and it is known that it converges to a Nash equilibrium in
several classes of games: two-player zero-sum games
\citep{robinson1951iterative}, $2\times n$ games
\citep{berger2005fictitious}, potential and identical-interest games
\citep{monderer1996fictitious}. This convergence
property is, however, not universal. There exist famous
counterexamples: the $2\times2$ instance of
\citet{monderer1996a2}, Shapley's cycling example
\citep{shapley1964some}, 
Jordan's
three-player matching-pennies variant \citep{jordan1993three} (see also ~\citep{FOSTER199879,gaunersdorfer1995fictitious}). 
The present paper adds an auction-theoretic counterexample to this list and demonstrates that the convergence of fictitious play hinges on the auction's tie-breaking rule.

The computation of equilibria in first-price auctions
\citep{krishna2009auction} has been researched for decades. A
theoretical characterization is available in the symmetric case when
values are drawn from a continuous distribution
\citep{vickrey1961counterspeculation,milgrom1982theory,maskin2003uniqueness}.
Seminal numerical methods rely on systems of ordinary differential
equations derived from the first-order conditions
\citep{marshall1994numerical,fibich2011numerical,fibich2012asymmetric},
but these are known to be numerically unstable; the discrete-value case
is treated separately in \citet{wang2020bayesian}. In parallel, it is of
independent interest to know whether \emph{learning} dynamics converge
to an equilibrium, both because such dynamics provide a behavioral
justification of equilibrium play and because they remain agnostic to
the structure of the value distribution. This question has acquired
practical urgency since the migration of display-advertising
marketplaces to first-price auctions \citep{paesleme2020why}, which
renewed interest in their computational properties
\citep{filosratsikas2021complexity,bichler2021learning}.
Recent contributions in the field of computing auction equilibrium include in particular the work of Bichler and co-authors~\citep{bichler2021learning,bichler2023learning,bichler2025computing,bichler2025beyond,ahunbay2025uniqueness} (differentiable approaches, no-regret approaches).

 The analysis is motivated by, and
complements, the empirical findings reported in
\citet{heymann2025empirical} (see also the earlier version
\cite{heymann2021heuristic}) for which the authors developed the open-source
library \texttt{FP4FPA} \citep{fp4fpa}.
 There, fictitious play was
observed to converge to an $\epsilon$-equilibrium across a broad set of
first-price auction instances, including instances with correlated
values, and the convergence  visibly required to replace the
standard tie-breaking rule with one that awards a payoff of zero to all
bidders in case of a tie.
Observe that a  methodology that guaranties convergence to an estimate Nash equilibrium in first-price auction with correlated values is yet to be found\footnote{A number of methods can identify a Nash equilibrium if they converge, yet convergence itself remains uncertain.}, and that the  empirical results from~\citet{heymann2025empirical}
give hope that the methodology developed in this paper could further extend to some  settings with correlated values. 
The present paper provides the theoretical
counterpart of those observations in the 2-bidder, symmetric, independent,
discrete setting. Our contributions are:
\begin{enumerate}[label=(\roman*)]
  \item a minimal, fully worked counterexample
    (Section~\ref{sec:ce}) showing that continuous-time fictitious play
    with the standard uniform-split rule fails to converge to the
    symmetric Bayes--Nash equilibrium, converging instead to a stable
    limit cycle;
  \item a convergence theorem (Section~\ref{sec:main}) showing that
    under the modified zero-on-tie rule, continuous-time fictitious play
    converges to a Nash equilibrium of the modified game. We then justify with Theorem~\ref{thm:epsilon-uniform} why this equilibrium is also
     $\epsilon$-equilibrium of the original auction in many settings.
\end{enumerate}

\section{First-price auction and fictitious play}
\label{sec:model}

For a finite set $X$ we write $\simplex(X)$ for the set of probability
distributions on $X$ and $\dirac{x}\in\simplex(X)$ for the unit mass at
$x\in X$. We write $\indic{\cdot}$ for the indicator of an event.

\subsection{Symmetric first-price auctions with independent discrete values}
\label{sec:symmetric}

An auction is a Bayesian game in which the players are the bidders. Each bidder observes a private valuation and submits a bid. The item is allocated to the highest bidder, who pays their own bid according to the first-price rule.

We focus on symmetric auctions with independent private values. This means that each bidder's valuation is drawn independently from a common distribution supported on the finite set
$0 \le \theta_1 < \cdots < \theta_m,
$
where value $\theta_j$ occurs with probability $\rho_j$. Bids are restricted to the finite grid
$
\bidset=\{\beta_1,\dots,\beta_K\}$, such that $
0\leq \beta_1<\cdots<\beta_K.
$

Since both the value and bid spaces are finite, a (mixed) bidding strategy is naturally represented as a collection of probability  distributions over admissible bids,
 one for each value:  $
\policy=(\policy_1,\dots,\policy_m)$ with
$
\policy_j\in\simplex(\bidset)
$ denotes the bidding policy associated with valuation $\theta_j$. 

We  deliberately restrict our attention to symmetric strategy profiles, in which every bidder uses the same bidding policy $\policy$.

\subsection{Two tie-breaking rules}
\label{sec:tiebreak}

We compare two tie-breaking rules, that clarify what happens when the two bidders have the same bid. 

\begin{definition}[Uniform tie-breaking]
The bidders   receive the item with probability $1/2$ and pay their bids. 
\end{definition}

\begin{definition}[Zero-on-tie]
\label{def-zero-on-tie}
The payoff is $0$ for each bidder.
\end{definition}

The uniform-split rule is the most standard convention. The zero-on-tie  is hard to justify as an auction
\emph{mechanism}, but we argue that it is a legitimate device for \emph{computing}
equilibria: as we show in Section~\ref{sec:main} it makes fictitious play
converge, and  the two games
have nearby equilibria when the bid grid is fine enough, so that the limit
is an $\epsilon$-equilibrium of the original auction.

\subsection{Continuous-time fictitious play}
\label{sec:fp}

We study fictitious play in continuous time. For a symmetric policy
$\policy=(\policy_1,\dots,\policy_m)$, define the best-reply
correspondence of value type $\theta_j$ by
\begin{align}
  \BR_j(\policy)
  &=
  \argmax_{b\in\bidset}
  \payoff_j(\dirac b,\policy),
  \label{eq:br}
\end{align}
where 
\begin{align}
\payoff_j(\mu,\policy)
=
\sum_{b\in\bidset}
\mu(b)
\left[
(\theta_j-b)\,
P_{\mathrm{win}}(b,\policy)
\right],
\end{align}
and  
$P_{\mathrm{win}}(b,\policy)$  is the probability that bid b wins against the opponents following the  policy $\policy$ (depends on the tie breaking rule).
Continuous-time fictitious play is the differential inclusion
\begin{align}
  \dot{\policy}_j
  &\in
  \overline{\operatorname{conv}}
  \BR_j(\policy)
  -
  \policy_j,
  \qquad
  j=1,\ldots,m,
  \label{eq:di}
\end{align}
where
$\overline{\operatorname{conv}}\BR_j(\policy)\subseteq\simplex(\bidset)$
denotes the set of probability distributions\footnote{this is notably different from the formalism in ~\citep{krishna1998convergence}, which does not allow for randomization} supported on
$\BR_j(\policy)$.

The right-hand side of \eqref{eq:di} is a nonempty, compact,
convex-valued, and upper semicontinuous correspondence. Standard results
on differential inclusions therefore guarantee the existence of
absolutely continuous solutions from every initial condition
\citep{aubinDifferentialInclusionsSetValued1984,benaim2005stochastic}, and we take
\eqref{eq:di} as the object of analysis throughout.

The differential inclusion \eqref{eq:di} is the continuous-time
counterpart of the  discrete-time fictitious play process. At
iteration $k$, each value type selects a best reply
$b_j^{(k)}\in\BR_j(\policy^{(k)})$ and updates its policy according to
$
  \policy_j^{(k+1)}
  =
  (1-\eta_k)\policy_j^{(k)}
  +
  \eta_k\,\dirac{b_j^{(k)}},
$, $
  \eta_k=\frac{1}{k+1}$.
Under the logarithmic time rescaling the interpolated
discrete trajectories track the solutions of \eqref{eq:di}.

\section{Failure of Convergence with Uniform Tie-Breaking}
\label{sec:ce}

\subsection{The auction and its equilibrium}
We take\footnote{we use a zero based indexation here as it is more convenient in this section}: 
The value is high, $\theta_1 = 1$ with probability $\rho_1=1/3$ and low $\theta_0=0$
with probability $\rho_0=2/3$. Bids belong to
$\bidset=\{0,0.1,0.2,0.3\}$, and we write
$\beta_0=0,\ \beta_1=0.1,\ \beta_2=0.2,\ \beta_3=0.3$. 
Low type $0$  ( obtains $0$ by bidding $\beta_0=0$ and a negative
payoff from any winning positive bid, so it bids $0$. 
In what follows, we denote by $x(t)=(x_0(t),x_1(t),x_2(t),x_3(t))\in\simplex(\bidset)$  the empirical
distribution of the high type over the four bids.
We first record that
the high type never bids $0$.

\begin{lemma}
\label{lem:nozero}
With uniform tie-breaking,  bid $\beta_0=0$ is never a best reply
for the high type:
$\beta_0\notin\overline{\operatorname{conv}}
  \BR_1(\policy(t))$ for almost all $t$.
\end{lemma}

\begin{proof}
 The opponent has value
$0$ with probability $2/3$ (then bids $0$) and value $1$ with probability
$1/3$ (then bids according to $x$). The unconditional opponent bid
probabilities for $0,0.1,0.2,0.3$ are respectively
$\tfrac23+\tfrac{x_0}{3}$,
  $\tfrac{x_1}{3}$,
  $\tfrac{x_2}{3}$, and
  $\tfrac{x_3}{3}$.
Hence, the high type's payoff from a bids $0$ and 1 are respectively\footnote{using  the formula
$(1-b)\big[\Prob(\text{opponent bids below }b)+\tfrac12\Prob(\text{opponent bids }b)\big]$}
$ \tfrac13+\tfrac{x_0}{6}<\tfrac13+\tfrac{1}{6}=\frac12$
  and  $\tfrac{9}{10}\Big(\tfrac23+\tfrac{x_0}{3}+\tfrac{x_1}{6}\Big)>\frac12$.
We conclude that the bid 0 is strictly dominated by the bid  0.1 for the high type.
\end{proof}

Consequently the best reply lies in $\{\beta_1,\beta_2,\beta_3\}$ and the
$x_0$-coordinate decays autonomously, $\dot x_0=-x_0$. We therefore study
the dynamics on the high-type spanned by $(x_1,x_2,x_3)$.

We can then derive  a \textbf{Nash equilibrium} for this game. 
\begin{proposition}
\label{prop:ce-eq}
The policy $\gamma^\star$ defined as 
$\gamma_0^\star =(1,0,0,0)$ and $ \gamma_1^\star =(0,4/13,3/13,6/13)$
is a Nash equilibrium with common supported interim payoff  for the high type $\tfrac{42}{65}\approx0.6462$.
\end{proposition}

\begin{proof}
Under this putative equilibrium, 
the unconditional opponent bid
probabilities for $0,0.1,0.2,0.3$ are respectively
$\tfrac23+\tfrac{\gamma_1^\star(0)}{3}=\tfrac23$,
  $\tfrac{\gamma_1^\star(1)}{3}=\tfrac{4}{39}$,
  $\tfrac{\gamma_1^\star(2)}{3}=\tfrac{3}{39}$, and
  $\tfrac{\gamma_1^\star(3)}{3}=\tfrac{6}{39}$. For a high type, the  payoff for bidding $0.1$ $0.2$ or $0.3$ is respectively
$0.9(2/3+0.5\tfrac{4}{39})=\tfrac{42}{65}$,
$0.8(2/3+\tfrac{4}{39}+0.5\tfrac{3}{39})=\tfrac{42}{65}$, and
$0.7(2/3+\tfrac{4}{39}+\tfrac{3}{39}+0.5\tfrac{6}{39})=\tfrac{42}{65}$.
\end{proof}
It is easy to show that this equilibrium is the unique symmetric equilibrium of the game. 
\subsection{Reduced dynamics and  Poincar\'e map}
Let $u_i$ denote the expected payoff of the high type from bidding
$\beta_i$, $i=1,2,3$. Since only the ordering of these payoffs matters
for best responses, we introduce the scaled payoff differences
\begin{align*}
A &= 60\,(u_1-u_2)=4+2x_0-7x_1-8x_2, &\\
B &= 60\,(u_2-u_3)=-3+9x_0+9x_1+x_2.
\end{align*}
And because $60\,(u_1-u_3)= 60\,(u_1-u_2) + 60\,(u_2-u_3)=A+B$, we can identify
the best-reply regions as follows: 
\begin{align*}
  \BR_1(\gamma(t))&=\{\beta_1 \}\iff A>0,\ A+B>0 \\
  \BR_1(\gamma(t))&=\{\beta_2\} \iff A< 0,\ B> 0\\
  \BR_1(\gamma(t))&=\{\beta_3 \}\iff B< 0,\ A+B< 0 \ .
\end{align*}
The regions are separated by the three switching planes
$A=0$, $B=0$, and $A+B=0$. Their common intersection in the
coordinates $(A,B,x_0)$ is the line $A=B=0$. On the invariant face
$x_0=0$, this intersection reduces to the equilibrium high-type policy
$\gamma_1^\star=(0,4/13,3/13,6/13)$.
In the coordinates $(A,B,x_0)$ the differential inclusion \eqref{eq:di}
becomes piecewise affine, with the $x_0$-direction purely contracting
($\dot x_0=-x_0$) and the $(A,B)$-dynamics independent of $x_0$:
\begin{equation}
  \begin{array}{lll}
  \BR_1(\gamma(t))=\{\beta_1 \}\implies & \dot A=-3-A, & \dot B=6-B,\\
   \BR_1(\gamma(t))=\{\beta_2\}\implies & \dot A=-4-A, & \dot B=-2-B,\\
 \BR_1(\gamma(t))=\{\beta_3 \}\implies & \dot A=4-A, & \dot B=-3-B\ .
  \end{array}
  \label{eq:pwa}
\end{equation}

The situation is pictured in Figure~\ref{fig:ce-theory}.
Let $\Sigma=\{A=0,\ B=s>0,\ x_0=c\}$ be the section between the red and green areas, where the trajectory
crosses from the $\beta_1$-region into the $\beta_2$-region. Integrating
\eqref{eq:pwa} successively through the regions $\beta_2\to\beta_3\to\beta_1$
gives the following:
\begin{itemize}
    \item  first,   from $(A,B)=(0,s)$ in the {\color{green}{green}} region where $\BR_1=\beta_2$, we have 
    \begin{itemize}
        \item $A(t)=-4+4e^{-t}$ and $B(t)=-2+(s+2)e^{-t}$;
        \item the next switch is at $B=0$, i.e.
$e^{-t_1}=2/(s+2)$, where $A=-r$ with $r=\tfrac{4s}{s+2}$;
    \end{itemize}
    \item     then, from $(A,B)=(-r,0)$ in the   {\color{blue}{blue}}  region where $\BR_1=\beta_3$, follows 
    \begin{itemize}
        \item $A(t)=4+(-r-4)e^{-t}$, and $B(t)=-3+3e^{-t}$;
        \item the next switch is at $A+B=0$,
i.e. $e^{-t_2}=1/(r+1)$, where $(A,B)=(q,-q)$ with $q=\tfrac{3r}{r+1}$;
\end{itemize}

    \item     finally, from $(A,B)=(q,-q)$ in the {\color{red}{red}}  region where $\BR_1=\beta_1$, follows 
    \begin{itemize}
        \item $A(t)=-3+(q+3)e^{-t}$ and $B(t)=6+(-q-6)e^{-t}$;
        \item the return to $\Sigma$ is at
$A=0$, i.e. $e^{-t_3}=3/(q+3)$, where $B'=\tfrac{3q}{q+3}$.
\end{itemize}
\end{itemize}

Composing $r=\tfrac{4s}{s+2}$, $q=\tfrac{3r}{r+1}$, and $B'=\tfrac{3q}{q+3}$
gives the  \textbf{Poincar\'e map}
\begin{equation}
P(s)=\dfrac{12s}{9s+2} .
  \label{eq:poincare}
\end{equation}
Including the contracting coordinate, the full return map on the
two-dimensional section is
\begin{equation}
  (s,c)\ \longmapsto\ \Big(\tfrac{12s}{9s+2},\ \tfrac{2}{9s+2}\,c\Big).
  \label{eq:returnmap}
\end{equation}

\begin{theorem}
\label{thm:ce}
For the instance above, the symmetric Bayes--Nash equilibrium $x^\star$ is
 unstable under continuous-time fictitious play with the
uniform-split rule, and the dynamics admit a stable periodic orbit.
\end{theorem}

\begin{proof}
Differentiating \eqref{eq:poincare}, $P'(s)=\tfrac{24}{(9s+2)^2}$, so
$P'(0)=24/4=6>1$: perturbations on the section are expanded, and the
equilibrium is  unstable. The fixed-point equation
$s=\tfrac{12s}{9s+2}$ has roots $s=0$ and $9s+2=12$, i.e.\ $s^\star=10/9$,
where $P'(s^\star)=24/12^2=1/6<1$. Moreover
\[
  P(s)-s=\frac{s(10-9s)}{9s+2},
\]
so $P(s)>s$ for $0<s<10/9$ and $P(s)<s$ for $s>10/9$: the map pushes
trajectories toward $s^\star$. The derivative of  $\tfrac{2}{9s+2}$ in
\eqref{eq:returnmap} equals $1/6$ at $s^\star$, so the full return map has
derivative $\operatorname{diag}(1/6,1/6)$ at $(s^\star,0)$ and the cycle is
attracting in both directions.
\end{proof}

Figure~\ref{fig:ce-sim} shows  a simulation on a discrete fictitious play.  

\begin{figure}[t]
  \centering
  \includegraphics[width=0.82\linewidth]{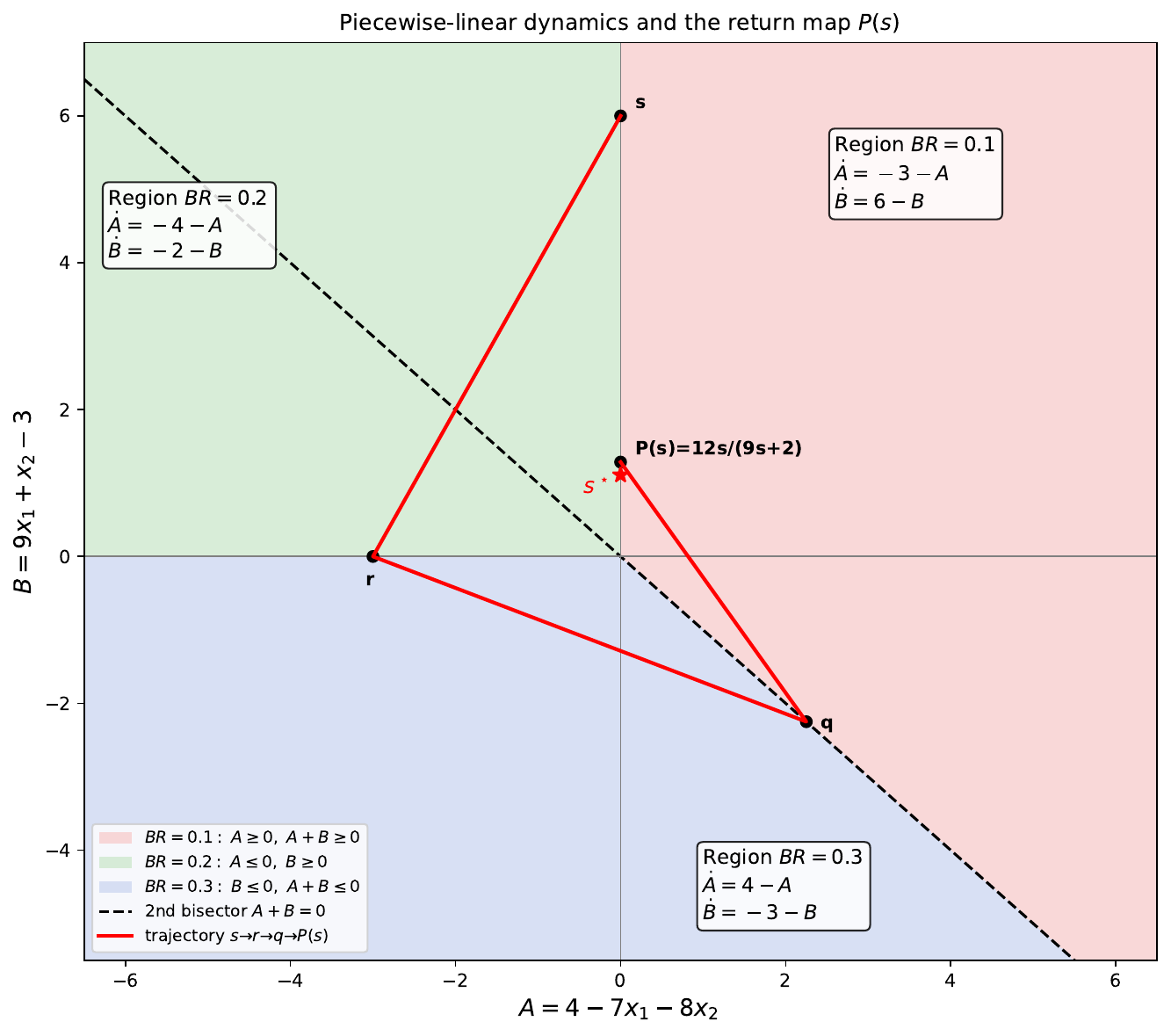}
  \caption{Piecewise-affine continuous-time fictitious-play dynamics in the
  $(A,B)$ plane (the invariant face $x_0=0$) for the $3$-bid
  counterexample under the uniform-split rule. The three shaded cells are
  the best-reply regions; the red path is one return loop
  $s\to r\to q\to P(s)$ through the section $\{A=0,\ B>0\}$. The unstable
  equilibrium sits at $A=B=0$ and the cycle fixed point at $s^\star=10/9$
  (red star). 
}
  \label{fig:ce-theory}
\end{figure}

\begin{figure}[t]
  \centering
  \begin{minipage}{0.49\linewidth}\centering
    \includegraphics[width=\linewidth]{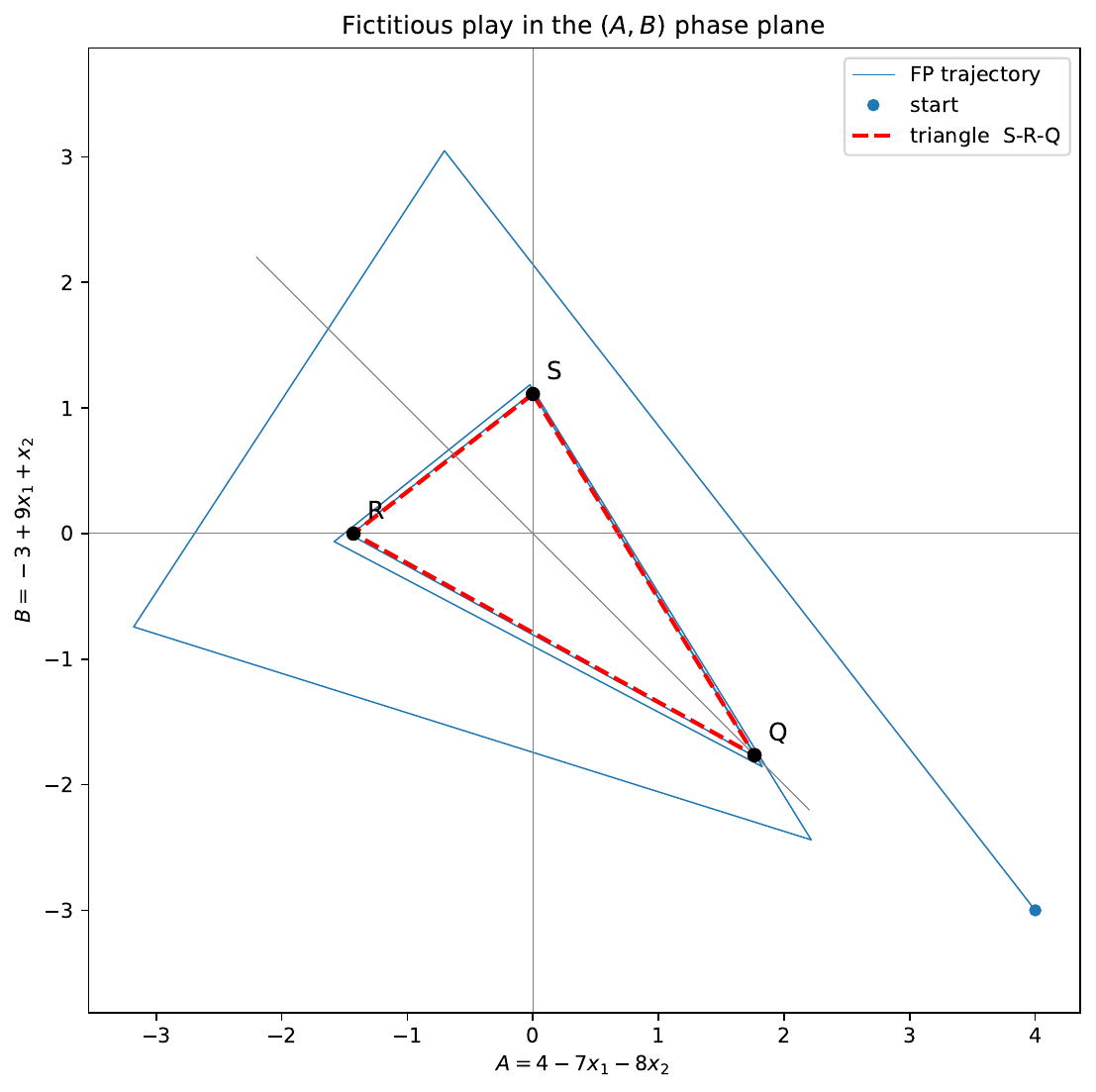}\\[-2pt]
    {\small (a) starting from a generic point}
  \end{minipage}\hfill
  \begin{minipage}{0.49\linewidth}\centering
    \includegraphics[width=\linewidth]{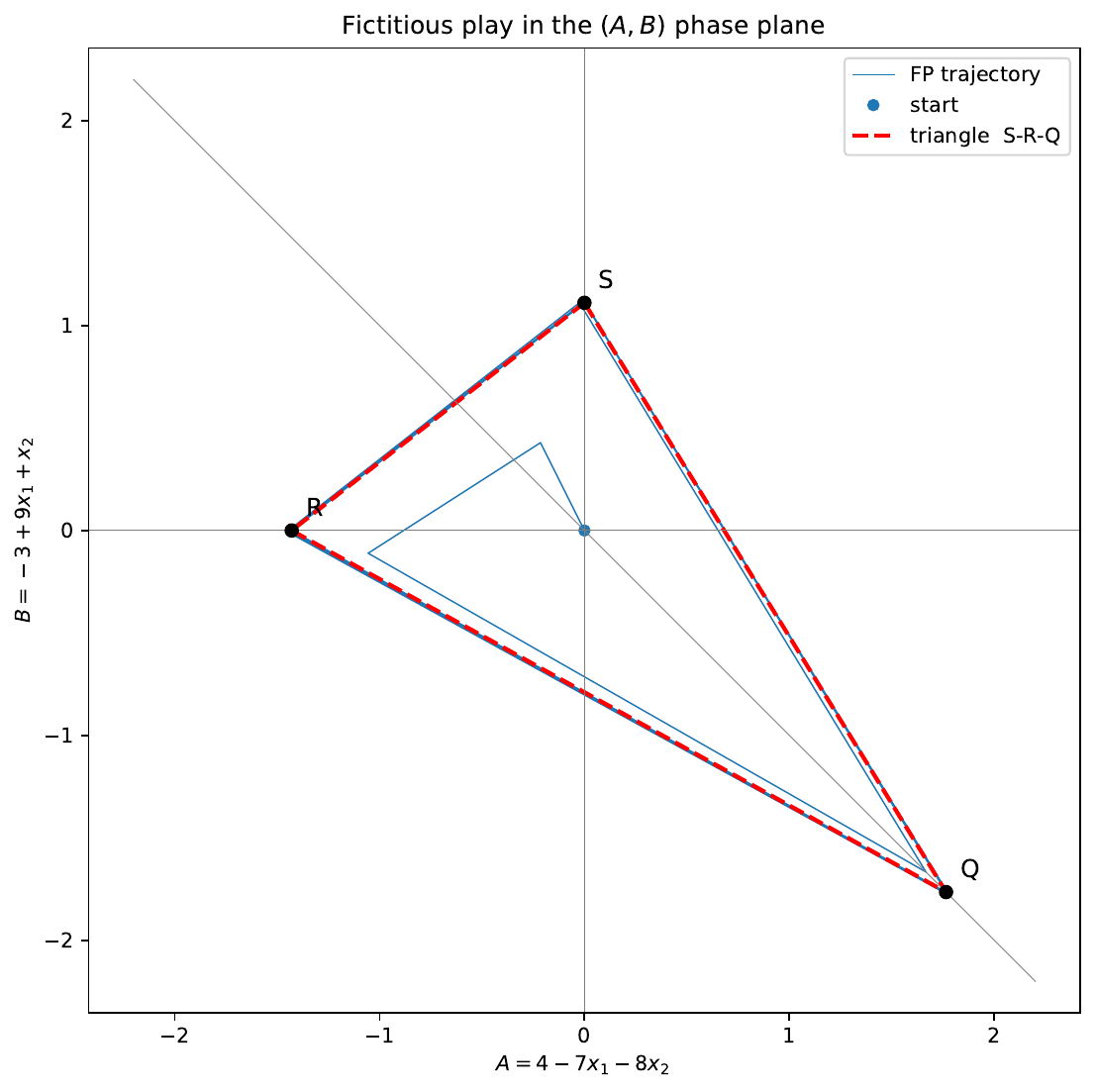}\\[-2pt]
    {\small (b) starting from the unstable equilibrium}
  \end{minipage}
  \caption{Discrete fictitious play on the $3$-bid counterexample with the
  uniform-split rule, in the $(A,B)$ plane. From a generic start (a) the
  iterates spiral inward and from the equilibrium (b) they spiral outward,
  both converging to the  cycle $S\to R\to Q$ of
  Figure~\ref{fig:ce-theory}. 
}
  \label{fig:ce-sim}
\end{figure}

\section{Convergence under the modified rule}
\label{sec:main}

\emph{Index conventions.}
Throughout this section $i,j,k\in\{1,\dots,m\}$ index \emph{value types} --- type $j$ has value $\theta_j$ and probability $\rho_j$ --- while $a,\ell,p,q,r,s\in\{1,\dots,K\}$ index \emph{bids} on the grid, so $\beta_s$ is the $s$-th bid and $\ell_j\le r_j$ are the block endpoints of type $j$.
On the cumulative quantities $\Gamma,\Psi,\Lambda$ the subscript is a value type and the superscript a bid threshold, e.g.\ $\Gamma_j^\ell$; the aggregate $\Phi^\ell=\sum_k\rho_k\Gamma_k^\ell$ sums over types and carries only the bid superscript.
A star marks an equilibrium quantity ($\gamma^\star$, $\Gamma_j^{\star,s}$) and a hat a selector quantity ($\widehat\Gamma_j^\ell$), while a lower-case $c$ with a bid subscript ($c_s$) is a coefficient, not an index.
Finally $\Gamma_j^\ell$ is a \emph{cumulative} lower tail, whereas $\mu_s$ and $x_j(\beta_s)$ are point masses at $\beta_s$.

Throughout this section the tie-breaking rule is zero-on-tie (Definition~\ref{def-zero-on-tie}).
We write $x_j\in\simplex(\bidset)$ for the empirical policy of type $j$.
For a threshold $\ell\in\{1,\dots,K+1\}$ define the \emph{(strict) lower tail}
\begin{align}
  \Gamma_j^\ell(x)=\sum_{q<\ell}x_j(\beta_q),
\end{align}
so $\Gamma_j^1(x)=0$ and $\Gamma_j^{K+1}(x)=1$.
Under zero-on-tie a bid wins only against strictly lower bids, so the payoff of type $j$ from bid $\beta_\ell$ is
\begin{align}
  \phi_{j,\ell}(x)=(\theta_j-\beta_\ell)\,\Phi^\ell(x),
  \qquad
  \Phi^\ell(x)=\sum_{k=1}^m\rho_k\Gamma_k^\ell(x),
  \label{eq:zero-payoff}
\end{align}
and the active value of type $j$ is $w_j(x)=\max_{\ell}\phi_{j,\ell}(x)$; thus $\BR_j(x)=\{\beta_\ell:\phi_{j,\ell}(x)=w_j(x)\}$.
Continuous-time fictitious play is $\dot x_j=\alpha_j-x_j$ with $\alpha_j(t)\in\overline{\operatorname{conv}}\BR_j(x(t))$; writing $\widehat\Gamma_j^\ell(t)=\sum_{q<\ell}\alpha_j(\beta_q,t)$ for the selector tail, summation of the dynamics gives
\begin{align}
  \dot\Gamma_j^\ell(t)=\widehat\Gamma_j^\ell(t)-\Gamma_j^\ell(t)
  \qquad\text{for a.e. }t.
  \label{eq:gamma-dynamics}
\end{align}
Finally we write $\Psi_j^\ell(t)=\sum_{k\neq j}\rho_k\Gamma_k^\ell(t)=\Phi^\ell-\rho_j\Gamma_j^\ell$ for the environment faced by type $j$, and $\mu_s(t)=\Phi^{s+1}(x(t))-\Phi^s(x(t))=\sum_k\rho_k x_k(\beta_s,t)$ for the total mass placed at bid $\beta_s$.

\begin{assumption}[Admissibility]
\label{ass:admissible}
Each type bids strictly below its value: for every $j$ the policy $x_j$ and every selector $\alpha_j$ are supported on the admissible set $\bidset_j=\{\beta_\ell:\beta_\ell<\theta_j\}$.
Equivalently, the best-reply correspondence \eqref{eq:br} maximises over $\bidset_j$ rather than over all of $\bidset$.
\end{assumption}

Bidding at or above one's value is weakly dominated under zero-on-tie: it yields a payoff $\le0$, whereas some bid below the value yields $\ge0$.
Assumption~\ref{ass:admissible} simply discards these dominated actions.
It is costless for our result --- the candidate equilibrium $\gamma^\star$ constructed below uses only bids strictly below each value --- but it is \emph{not} cosmetic:

\begin{remark}[Over-bidding equilibria]
\label{rem:overbid}
Without Assumption~\ref{ass:admissible} the conclusion of Theorem~\ref{thm:main} fails.
Take a single value type $\theta_1=\tfrac12$ ($m=1$, $\rho_1=1$) and bids $\bidset=\{0,1\}$.
Writing $x=(x_1,x_2)$ for the mass on $0$ and $1$, one has $\phi_{1,1}=(\tfrac12-0)\cdot0=0$ and $\phi_{1,2}=(\tfrac12-1)x_1=-\tfrac12 x_1$, so the construction below gives $\gamma^\star=\dirac{\beta_1}$.
Yet the constant profile $x(t)\equiv\dirac{\beta_2}$ solves the differential inclusion \eqref{eq:di}: there $\overline{\operatorname{conv}}\BR_1(x)-x=\simplex(\bidset)-\dirac{\beta_2}\ni0$, so $\dot x\equiv0$ is admissible, and this solution never reaches $\gamma^\star$.
Such rest points --- every bidder placing the same over-value bid and tying for a payoff of zero --- exist precisely because zero-on-tie makes over-bidding into a tie costless; they are weakly dominated and are removed by Assumption~\ref{ass:admissible}.
\end{remark}
\begin{lemma}[Monotonicity of best replies]
\label{lem:monotone-zero}
If $\beta_a<\beta_s\in\BR_i(x)$ and $i<j$, then $\phi_{j,s}(x)\ge\phi_{j,a}(x)$, with strict inequality when $\Phi^s(x)>\Phi^a(x)$.
\end{lemma}

\begin{proof}
Since $\phi_{j,\ell}=(\theta_j-\beta_\ell)\Phi^\ell$,
\begin{align*}
  \phi_{j,s}(x)-\phi_{j,a}(x)
  =\bigl[\phi_{i,s}(x)-\phi_{i,a}(x)\bigr]+(\theta_j-\theta_i)\bigl(\Phi^s(x)-\Phi^a(x)\bigr).
\end{align*}
The bracket is $\ge0$ because $\beta_s\in\BR_i(x)$; the second term is $\ge0$ because $\theta_j>\theta_i$ and $\Phi^s\ge\Phi^a$ (as $a<s$), and is strictly positive when $\Phi^s>\Phi^a$.
\end{proof}
\subsection{The candidate equilibrium}

We define a putative equilibrium recursively.
The lowest type is pure at the highest bid below $\theta_1$,
\begin{align*}
  \gamma_1^\star=\dirac{\beta_{r_1}},\qquad r_1=\max\{\ell:\beta_\ell<\theta_1\}.
\end{align*}
We also set $\ell_1=r_1$, $S_1={\beta_{r_1}}$, and $\pi_1^\star=0$.
Suppose $\gamma_1^\star,\dots,\gamma_{j-1}^\star$ are constructed.
Write $\Gamma_k^{\star,s}=\Gamma_k^s(\gamma^\star)$ and define the lower-type contribution
\begin{align}
  \Lambda_j^\ell=\sum_{k<j}\rho_k\Gamma_k^{\star,\ell}.
\end{align}
Take $\ell_j$ to be the smallest maximizer of $(\theta_j-\beta_\ell)\Lambda_j^\ell$ over $\{\ell:\beta_\ell<\theta_j\}$, and set $\pi_j^\star=(\theta_j-\beta_{\ell_j})\Lambda_j^{\ell_j}$.
The remaining masses of $\gamma_j^\star$ are chosen so that all bids in a consecutive block $S_j=\{\beta_{\ell_j},\dots,\beta_{r_j}\}$ yield the same payoff $\pi_j^\star$:
\begin{align}
  (\theta_j-\beta_s)\bigl(\Lambda_j^s+\rho_j\Gamma_j^{\star,s}\bigr)=\pi_j^\star,
  \qquad s=\ell_j,\dots,r_j,
  \label{eq:support-equality}
\end{align}
where $r_j$ is the largest index for which \eqref{eq:support-equality} admits a solution $\Gamma_j^{\star,s}\in[0,1]$ with $\beta_s<\theta_j$; below the block $\Gamma_j^{\star,s}=0$ and above it $\Gamma_j^{\star,s}=1$.
From Lemma~\ref{lem:monotone-zero} the construction produces weakly ordered support blocks,
\begin{align}
  r_j\le \ell_{j+1},\qquad j=1,\dots,m-1.
  \label{eq:block-order}
\end{align}

\begin{lemma}
\label{lem:gamma-nash}
The profile $\gamma^\star$ is a Nash equilibrium of the zero-on-tie auction.
\end{lemma}

\begin{proof}
Fix type $j$.
By construction every bid in $S_j$ gives payoff $\pi_j^\star$; by the choice of $\ell_j$ no lower bid does better.
If $r_j=\ell_{j+1}$, then, for $q\in S_{j+1}$
\begin{align*}
  (\theta_j-\beta_{r_j})\Phi^{\star,r_j} - (\theta_j-\beta_{q})\Phi^{\star,q}=(\Phi^{\star,r_j}-\Phi^{\star,q})(\theta_j-\beta_{r_j})+\Phi^{\star,q}(\beta_q-\beta_{r_j})\\
  \geq (\Phi^{\star,r_j}-\Phi^{\star,q})(\theta_{j+1}-\beta_{r_j})+\Phi^{\star,q}(\beta_q-\beta_{r_j})
  = 0. 
\end{align*}
If $r_j+1=\ell_{j+1}$, then by construction, $\Phi^{\star,r_{j+1}}(\theta_j-\beta_{r_{j+1}})\leq   (\theta_j-\beta_{r_j})\Phi^{\star,r_j}$ and the same reasoning applies. 
Hence by a simple recursion, no higher bid does strictly better.
By \eqref{eq:block-order} the higher types place no mass below $\beta_{r_j}$, so the payoff of type $j$ against $\gamma^\star$ is $(\theta_j-\beta_s)(\Lambda_j^s+\rho_j\Gamma_j^{\star,s})$ on $S_j$.
Hence every best reply of type $j$ lies in $S_j$ and $\gamma_j^\star$ is a best response; this holds for all $j$.
\end{proof}

\subsection{Auxiliary lemmas}

\begin{lemma}
\label{lem:flat-tie}
If $\beta_s\in\BR_j(x(t))$ for a.e.\ $t>T$, then for every $k>j$ and every $a<s$, $\alpha_k(\beta_a,t)=0$ for a.e.\ $t>T$.
\end{lemma}

\begin{proof}
Fix $k>j$ and $a<s$.
By Lemma~\ref{lem:monotone-zero}, $\phi_{k,s}\ge\phi_{k,a}$.
If $\alpha_k(\beta_a,t)>0$ then $\beta_a\in\BR_k(x(t))$, so $\phi_{k,a}\ge\phi_{k,s}$; hence $\phi_{k,s}=\phi_{k,a}$, and by the strict clause of Lemma~\ref{lem:monotone-zero} this forces $\Phi^s(x(t))=\Phi^a(x(t))$, i.e.\ no mass on $\{\beta_a,\dots,\beta_{s-1}\}$; in particular $x_k(\beta_a,t)=0$.

As $t\mapsto x_k(\beta_a,t)$ is absolutely continuous and nonnegative, it has zero derivative a.e.\ on its zero set, so $\dot x_k(\beta_a,t)=0$; with $\dot x_k(\beta_a,t)=\alpha_k(\beta_a,t)-x_k(\beta_a,t)$ this gives $\alpha_k(\beta_a,t)=0$.
\end{proof}

\begin{lemma}
\label{lem:gamma0}
Under the hypothesis of Lemma~\ref{lem:flat-tie}, $\Gamma_k^s(t)\to0$ for every $k>j$.
\end{lemma}

\begin{proof}
By Lemma~\ref{lem:flat-tie}, $\widehat\Gamma_k^s(t)=\sum_{a<s}\alpha_k(\beta_a,t)=0$ for a.e.\ $t>T$; then \eqref{eq:gamma-dynamics} gives $\dot\Gamma_k^s=-\Gamma_k^s$, so $\Gamma_k^s(t)\to0$.
\end{proof}

\begin{lemma}
\label{lem:lower-endpoint}
Fix $j$ with $\pi_j^\star>0$ (so $\ell_j$ is the \emph{strict} maximizer of $(\theta_j-\beta_\ell)\Lambda_j^\ell$) and suppose $x_k(t)\to\gamma_k^\star$ for every $k<j$.
Then, for all $t$ large, $\beta_{\ell_j}\in\BR_j(x(t))$ and $\alpha_j(\beta_a,t)=0$ for every $a<\ell_j$.
Consequently $\Gamma_j^{\ell_j}(t)\to0$, and $\Gamma_k^{\ell_j}(t)\to0$ for every $k>j$.
\end{lemma}

\begin{proof}
Let $\beta_p$ be the smallest bid that lies in $\BR_j(x(t))$ for arbitrarily large $t$.
Bids below $\beta_p$ are not optimal for $t$ large, so $\alpha_j(\beta_a,t)=0$ for $a<p$ and thus $\Gamma_j^{p}(t)\to0$; by Lemma~\ref{lem:gamma0} (applicable at $\beta_p$) $\Gamma_k^{p}(t)\to0$ for $k>j$.

Since no mass sits strictly below $\beta_p$ --- neither from type $j$ ($\Gamma_j^p\to0$) nor from the higher types ($\Gamma_k^p\to0$, $k>j$) --- only the lower types contribute, so $\Phi^b(t)\to\sum_{k<j}\rho_k\Gamma_k^{\star,b}=\Lambda_j^b$ and hence $\phi_{j,b}(x(t))\to(\theta_j-\beta_b)\Lambda_j^b$ for every $b\le p$.
As $\beta_p$ is optimal at arbitrarily large times, $\phi_{j,p}=w_j\ge\phi_{j,b}$ there for every $b$; passing to the limit over such times gives $(\theta_j-\beta_p)\Lambda_j^p\ge(\theta_j-\beta_b)\Lambda_j^b$ for all $b\le p$.
Because $\ell_j$ is the strict maximizer, $p>\ell_j$ would violate this (taking $b=\ell_j$), while $p<\ell_j$ is excluded since $w_j\ge\phi_{j,\ell_j}$ together with $\liminf_t\Phi^{\ell_j}\ge\sum_{k<j}\rho_k\Gamma_k^{\star,\ell_j}=\Lambda_j^{\ell_j}$ forces $\liminf_t w_j\ge(\theta_j-\beta_{\ell_j})\Lambda_j^{\ell_j}=\pi_j^\star$, above the limit $(\theta_j-\beta_p)\Lambda_j^p$ of $w_j$ along those times.
Hence $p=\ell_j$, i.e.\ $\beta_{\ell_j}\in\BR_j(x(t))$ for $t$ large.

The remaining claims follow.
\end{proof}

The induction runs on the types in increasing order, under the hypothesis
\begin{align}
  H(j):\qquad x_k(t)\to\gamma_k^\star\ \text{ and }\ \dot x_k(t)\to0
  \quad\text{for every }k<j.
  \tag{$H(j)$}
  \label{eq:IH}
\end{align}
Fix $j$ and assume $H(j)$; abbreviate $\ell=\ell_j$, $r=r_j$, $\rho=\rho_j$, and $\delta=1-\Gamma_j^{\star,r}=\gamma_j^\star(\beta_r)>0$, so $\Gamma_j^{\star,s}\le1-\delta$ for $s\le r$.

\begin{corollary}
\label{cor:endpoints}
Under $H(j)$: for $t$ large $\beta_\ell\in\BR_j(x(t))$ and $\widehat\Gamma_j^\ell(t)=0$, hence $\Gamma_j^\ell(t)\to0$; and $\Gamma_k^\ell(t)\to0$ for every $k>j$.
\end{corollary}

\begin{proof}
Lemma~\ref{lem:lower-endpoint} gives $\beta_\ell\in\BR_j(x(t))$ and $\alpha_j(\beta_a,t)=0$ for $a<\ell$, so $\widehat\Gamma_j^\ell=0$ and, by \eqref{eq:gamma-dynamics}, $\dot\Gamma_j^\ell=-\Gamma_j^\ell\to0$; the higher-type claim is Lemma~\ref{lem:gamma0} at $\beta_\ell$.
\end{proof}

\begin{lemma}
\label{lem:drift}
Under $H(j)$, for every $s$ with $\beta_s\in\BR_j(x(t))$ for all $t$ large, $\dot\Psi_j^s(t)\to0$; in particular $\dot\Psi_j^\ell(t)\to0$.
\end{lemma}

\begin{proof}
Write $\dot\Psi_j^s=\sum_{k<j}\rho_k\dot\Gamma_k^s+\sum_{k>j}\rho_k\dot\Gamma_k^s$.
For $k<j$, $\dot\Gamma_k^s\to0$ by the velocity clause of $H(j)$.
For $k>j$, Lemma~\ref{lem:flat-tie} gives $\widehat\Gamma_k^s=0$ and Lemma~\ref{lem:gamma0} gives $\Gamma_k^s\to0$, so $\dot\Gamma_k^s=-\Gamma_k^s\to0$ by \eqref{eq:gamma-dynamics}.
The case $s=\ell$ uses Corollary~\ref{cor:endpoints}.
\end{proof}

\subsection{The equilibrium selector and the value-gap identity}

For $s\ge\ell$ with $\beta_s<\theta_j$ set $c_s=\dfrac{\theta_j-\beta_\ell}{\theta_j-\beta_s}$ (so $c_\ell=1$, and $s\mapsto c_s$ increasing), and
\begin{align}
  \widehat\Gamma_j^{\star,s}(t):=\Gamma_j^s(t)-c_s\,\Gamma_j^\ell(t)
  +\frac{c_s\,\dot\Psi_j^\ell(t)-\dot\Psi_j^s(t)}{\rho}.
  \label{eq:selector-cum}
\end{align}
This is the selector tail that keeps the indifference $\phi_{j,s}=\phi_{j,\ell}$ invariant: differentiating $(\theta_j-\beta_s)\Phi^s=(\theta_j-\beta_\ell)\Phi^\ell$ and substituting $\dot\Gamma_j^\ell=-\Gamma_j^\ell$ (Corollary~\ref{cor:endpoints}) and $\widehat\Gamma_j^s=\Gamma_j^s+\dot\Gamma_j^s$ yields \eqref{eq:selector-cum}.
Define the value gap $D_{j,s}(t)=\phi_{j,s}(x(t))-\phi_{j,\ell}(x(t))$.

\begin{lemma}[Value-gap identity]
\label{lem:gap}
Under $H(j)$ there is $T$ such that, for every $s$ with $\ell\le s$, $\beta_s<\theta_j$, and $t>T$: $D_{j,s}(t)\le0$, with equality iff $\beta_s\in\BR_j(x(t))$, and
\begin{align}
  \dot D_{j,s}(t)=\rho\,(\theta_j-\beta_s)\bigl(\widehat\Gamma_j^s(t)-\widehat\Gamma_j^{\star,s}(t)\bigr)
  \qquad\text{for a.e. }t>T.
  \label{eq:gap-id}
\end{align}
\end{lemma}

\begin{proof}
For $t$ large $\beta_\ell\in\BR_j$ (Corollary~\ref{cor:endpoints}), so $\phi_{j,\ell}=w_j$; hence $D_{j,s}=\phi_{j,s}-w_j\le0$, with equality iff $\beta_s$ is optimal.
For the identity, start from $D_{j,s}=(\theta_j-\beta_s)\Phi^s-(\theta_j-\beta_\ell)\Phi^\ell$ and compute:
\begin{align*}
  \dot D_{j,s}
  &=(\theta_j-\beta_s)\,\dot\Phi^s-(\theta_j-\beta_\ell)\,\dot\Phi^\ell \\
  \intertext{Substitute $\dot\Phi^s=\dot\Psi_j^s+\rho(\widehat\Gamma_j^s-\Gamma_j^s)$ from \eqref{eq:gamma-dynamics} (and the same at $\ell$), then use $\widehat\Gamma_j^\ell=0$ (Corollary~\ref{cor:endpoints}), so that $\dot\Phi^\ell=\dot\Psi_j^\ell-\rho\Gamma_j^\ell$:}
  &=(\theta_j-\beta_s)\bigl[\dot\Psi_j^s+\rho(\widehat\Gamma_j^s-\Gamma_j^s)\bigr]
    -(\theta_j-\beta_\ell)\bigl[\dot\Psi_j^\ell-\rho\Gamma_j^\ell\bigr] \\
  \intertext{Split off the term $\rho(\theta_j-\beta_s)\widehat\Gamma_j^s$ and collect the rest into a bracket:}
  &=\rho(\theta_j-\beta_s)\widehat\Gamma_j^s
    -\Bigl[\rho(\theta_j-\beta_s)\Gamma_j^s-\rho(\theta_j-\beta_\ell)\Gamma_j^\ell
    +(\theta_j-\beta_\ell)\dot\Psi_j^\ell-(\theta_j-\beta_s)\dot\Psi_j^s\Bigr] \\
  \intertext{The bracket is exactly $\rho(\theta_j-\beta_s)\widehat\Gamma_j^{\star,s}$: multiply \eqref{eq:selector-cum} by $\rho(\theta_j-\beta_s)$ and use $(\theta_j-\beta_s)c_s=\theta_j-\beta_\ell$. Therefore}
  &=\rho(\theta_j-\beta_s)\bigl(\widehat\Gamma_j^s-\widehat\Gamma_j^{\star,s}\bigr),
\end{align*}
which is \eqref{eq:gap-id}.
\end{proof}

\begin{lemma}
\label{lem:coincide}
Under $H(j)$, for a.e.\ $t$ large and every $s$ with $\beta_s\in\BR_j(x(t))$ and $\beta_s<\theta_j$, $\widehat\Gamma_j^s(t)=\widehat\Gamma_j^{\star,s}(t)$.
\end{lemma}

\begin{proof}
$z_s:=-D_{j,s}$ is absolutely continuous, nonnegative, and vanishes on $A_s=\{t:\beta_s\in\BR_j(x(t))\}$, so $\dot z_s=0$ a.e.\ on $A_s$; with \eqref{eq:gap-id} and $\theta_j-\beta_s>0$ this gives the claim.
\end{proof}

\begin{remark}
\label{rem:admissible}
On the active block, Lemma~\ref{lem:coincide} identifies $\alpha_j^\star$ (the selector with tail \eqref{eq:selector-cum}) with the genuine selector $\alpha_j$, so its increments are automatically nonnegative there; off the block \eqref{eq:selector-cum} is only an auxiliary quantity.
\end{remark}

\subsection{The active set}

This subsection shows that under $H(j)$ the best replies of type $j$ settle to the equilibrium block: $\BR_j(x(t))=\{\beta_\ell,\dots,\beta_r\}$ for $t$ large (Lemma~\ref{lem:rconv}).
The argument has two directions.
For $\supseteq$, every bid of the block becomes optimal (a fill-in induction) and then stays optimal (persistence).
For $\subseteq$, no bid above $\beta_r$ can survive, because the cumulative level that would make it optimal exceeds $1$ and is therefore unreachable.
Four tools feed this: Lemma~\ref{lem:anchor} fixes the limiting payoff $\pi_j^\star$ and environment $\Lambda_j^s$; Lemma~\ref{lem:optlimit} identifies the limit $g_s$ of $\Gamma_j^s$ at an optimal bid; Lemma~\ref{lem:persist} shows that optimality, once gained with room to spare, is retained; and Lemma~\ref{lem:unfed} shows that trapped bids receive no mass, so gaps close.

\begin{lemma}
\label{lem:anchor}
Under $H(j)$, $\Psi_j^s(t)\to\Lambda_j^s$ for every $s$ optimal for all $t$ large, and $\phi_{j,\ell}(x(t))\to\pi_j^\star$.
\end{lemma}

\begin{proof}
$\Psi_j^s=\sum_{k<j}\rho_k\Gamma_k^s+\sum_{k>j}\rho_k\Gamma_k^s$; the first sum $\to\Lambda_j^s$ by $H(j)$, the second $\to0$ by Lemma~\ref{lem:gamma0}.
For $s=\ell$, $\phi_{j,\ell}=(\theta_j-\beta_\ell)(\Psi_j^\ell+\rho\Gamma_j^\ell)\to(\theta_j-\beta_\ell)\Lambda_j^\ell=\pi_j^\star$ (Corollary~\ref{cor:endpoints}).
\end{proof}
In what follows, $g_s$ is the cumulative level of $\Gamma_j^s$ at which $\beta_s$ would tie $\beta_\ell$ in payoff: a feasible level ($g_s\le1-\delta$) marks a bid inside the block, an infeasible one ($g_s>1$) a bid above it --- this is what pins the top of the active set at $r$.
\begin{lemma}
\label{lem:optlimit}
Under $H(j)$, let $\beta_s<\theta_j$ be optimal for a.e.\ $t$ large.
Then $\Gamma_j^s(t)\to g_s:=\tfrac1\rho\bigl(\tfrac{\pi_j^\star}{\theta_j-\beta_s}-\Lambda_j^s\bigr)$, with $g_s=\Gamma_j^{\star,s}\le1-\delta$ if $\ell\le s\le r$ and $g_s>1$ if $s>r$.
\end{lemma}

\begin{proof}
Indifference $\phi_{j,s}=\phi_{j,\ell}$ gives $(\theta_j-\beta_s)(\Psi_j^s+\rho\Gamma_j^s)=\phi_{j,\ell}$; letting $t\to\infty$ with $\Psi_j^s\to\Lambda_j^s$, $\phi_{j,\ell}\to\pi_j^\star$ (Lemma~\ref{lem:anchor}) yields $\Gamma_j^s\to g_s$.
For $\ell\le s\le r$, $g_s$ solves \eqref{eq:support-equality}, so $g_s=\Gamma_j^{\star,s}\le1-\delta$; for $s>r$ with $\beta_s<\theta_j$, \eqref{eq:support-equality} has no solution in $[0,1]$ (maximality of $r$) and the equalizing level is $>1$.
\end{proof}

\begin{lemma}
\label{lem:persist}
Under $H(j)$, let $\beta_q<\theta_j$ with $\Psi_j^q(t)$ convergent.
For every $\eta>0$ there is $T$ such that: if $\beta_q\in\BR_j(x(t_0))$ for some $t_0\ge T$ with $\Gamma_j^q(t_0)\le1-\eta$, then $\beta_q\in\BR_j(x(t))$ and $\Gamma_j^q(t)<1$ for all $t\ge t_0$.
\end{lemma}

\begin{proof}
Fix $\eta>0$.
While $\beta_q$ is optimal, Lemma~\ref{lem:coincide} gives $\widehat\Gamma_j^q=\widehat\Gamma_j^{\star,q}$, so \eqref{eq:gamma-dynamics} and \eqref{eq:selector-cum} yield successively        $\dot\Gamma_j^q=\widehat\Gamma_j^{\star,q}-\Gamma_j^q=-c_q\Gamma_j^\ell+\rho^{-1}(c_q\dot\Psi_j^\ell-\dot\Psi_j^q)$.
Integrating from $t_0$ to $t$,
\begin{align*}
  \Gamma_j^q(t)-\Gamma_j^q(t_0)
  &=\int_{t_0}^t\dot\Gamma_j^q\,ds \\
  &=-c_q\underbrace{\int_{t_0}^t\Gamma_j^\ell\,ds}_{\le\,\Gamma_j^\ell(t_0)\ \to\ 0}
    +\frac{c_q}{\rho}\underbrace{\bigl(\Psi_j^\ell(t)-\Psi_j^\ell(t_0)\bigr)}_{\to\,0}
    -\frac1\rho\underbrace{\bigl(\Psi_j^q(t)-\Psi_j^q(t_0)\bigr)}_{\to\,0},
\end{align*}
where the first term uses $\dot\Gamma_j^\ell=-\Gamma_j^\ell$ (Corollary~\ref{cor:endpoints}), so $\Gamma_j^\ell$ decays exponentially and $\int_{t_0}^t\Gamma_j^\ell\le\Gamma_j^\ell(t_0)$, while the last two vanish because $\Psi_j^\ell$ and $\Psi_j^q$ converge.
Hence, once $t_0$ is large, the right-hand side is $<\eta/2$ in absolute value for all $t\ge t_0$, so
\begin{align*}
  \Gamma_j^q(t)<\Gamma_j^q(t_0)+\tfrac\eta2\le(1-\eta)+\tfrac\eta2=1-\tfrac\eta2<1 .
\end{align*}
By \eqref{eq:selector-cum} this keeps $\widehat\Gamma_j^{\star,q}<1$, so the selector retains positive mass on $\beta_q$ and $D_{j,q}=0$ is maintained.
The set of $t\ge t_0$ at which $\beta_q$ is optimal is therefore relatively open and closed in $[t_0,\infty)$, hence all of it.
\end{proof}

\begin{lemma}[Unfed trapped bids; gaps close]
\label{lem:unfed}
Call $\beta_s$ \emph{trapped} at time $t$ if $\beta_s\notin\BR_j(x(t))$ while $\BR_j(x(t))$ contains a bid below and a bid above $\beta_s$.
Under $H(j)$, for a.e.\ large $t$:
\begin{enumerate}[label=(\roman*)]
  \item\emph{(unfed)} a trapped bid receives no mass: $\alpha_k(\beta_s,t)=0$ for every type $k$;
  \item\emph{(gaps close)} consequently a gap between two optimal bids cannot persist.
\end{enumerate}
\end{lemma}

\begin{proof}
\emph{(i)} Let $\beta_a<\beta_s<\beta_p$ with $\beta_a,\beta_p\in\BR_j(x(t))$ and $\beta_s\notin\BR_j(x(t))$; optimality of $\beta_p$ over $\beta_s$ gives $(\theta_j-\beta_p)\Phi^p>(\theta_j-\beta_s)\Phi^s$, hence $\Phi^p>\Phi^s$.
No type feeds $\beta_s$:
\begin{itemize}[label=$\bullet$]
  \item type $j$: $\beta_s\notin\BR_j$, so $\alpha_j(\beta_s,t)=0$;
  \item type $k<j$: were $\beta_s\in\BR_k$, then Lemma~\ref{lem:monotone-zero} applied to $\beta_a<\beta_s$ would give $\phi_{j,s}\ge\phi_{j,a}=w_j$ (as $\beta_a\in\BR_j$), i.e.\ $\beta_s\in\BR_j$ --- impossible;
  \item type $k>j$: from $\Phi^p>\Phi^s$ and the strict form of Lemma~\ref{lem:monotone-zero} ($\beta_s<\beta_p\in\BR_j$), $\phi_{k,p}>\phi_{k,s}$, so $\beta_s\notin\BR_k$.
\end{itemize}
Hence $\alpha_k(\beta_s,t)=0$ for every $k$.

\emph{(ii)} a gap between two optimal bids cannot persist because  the mass it encloses drains, while optimality of its upper endpoint keeps that mass bounded away from $0$, so the endpoint must leave $\BR_j$. Let $\beta_a<\beta_p$ be optimal with every bid between them non-optimal, hence trapped; by~(i) the enclosed mass $W=\Phi^p-\Phi^{a+1}$ is unfed, so \eqref{eq:gamma-dynamics} gives $\dot W=-W$, and $W$ decays exponentially while the gap lasts.
But optimality of $\beta_p$ over $\beta_{p-1}$ forces $\mu_{p-1}>\tfrac{\beta_p-\beta_{p-1}}{\theta_j-\beta_p}\Phi^{p-1}$.
Hence the gap closes: $\beta_p$ leaves $\BR_j$.
\end{proof}

\begin{lemma}
\label{lem:rconv}
Under $H(j)$, for $t$ large the active set is exactly the block, $\BR_j(x(t))=\{\beta_\ell,\dots,\beta_r\}$, and the mass strictly above it vanishes, $1-\Gamma_j^{r+1}(t)\to0$.
\end{lemma}
\begin{proof}
\begin{itemize}
    \item We first prove that every bid in the equilibrium block
$\{\beta_\ell,\ldots,\beta_r\}$ eventually becomes optimal, and this is done by induction on $q=\ell,\ldots,r$ after observing that 
the case $q=\ell$ is Corollary~\ref{cor:endpoints}.
Hence, assume  that
$\beta_\ell,\ldots,\beta_{q-1}$ are optimal for all sufficiently large
$t$.

Since $\beta_{q-1}$ is optimal,
Lemmas~\ref{lem:anchor} and~\ref{lem:optlimit} give respectively
$
\phi_{j,q-1}(t)\to\pi_j^\star,
$ (and hence
$
\Phi^{q-1}(t)
\to
\frac{\pi_j^\star}{\theta_j-\beta_{q-1}}
$) and $
\Gamma_j^{q-1}(t)\to g_{q-1}=\Gamma_j^{\star,q-1}
$.

\begin{itemize}
    \item  Ad absurdum $\beta_q$ cannot be \textbf{never} optimal after some time. Indeed, suppose that $\beta_q$ is \textbf{never} optimal after some time. Then by Lemma~\ref{lem:unfed} (ii), after some times, $\alpha_j(\beta_k,t)=0$ for $k\geq q$ but this implies that $\Gamma_j^q(t)>\Gamma_{\star,j}^q(t)$, which implies that $\phi_{j,q}>\phi_{j,l}$: contradiction. Hence, after any time $t_0$ there is another time $t$ where $\beta_q\in\BR_j$.
\item 
Whenever $\beta_q$ is optimal,
Lemma~\ref{lem:optlimit} gives
$\Gamma_j^q\to g_q$.
Since $q\le r$,
$g_q=\Gamma_j^{\star,q}\le1-\delta$.
Moreover, while $\beta_q$ is optimal,
higher types cannot place mass below $\beta_q$
(Lemma~\ref{lem:monotone-zero}),
so $\Psi_j^q$ converges.
Lemma~\ref{lem:persist} therefore applies and shows that,
once $\beta_q$ becomes optimal sufficiently late,
it remains optimal forever.

\item This completes the induction and proves that
$
\{\beta_\ell,\ldots,\beta_r\}
\subseteq
\BR_j(x(t))
$
for all sufficiently large $t$.
\end{itemize}
\item It remains to exclude bids above $\beta_r$.
Let $q>r$ with $\beta_q<\theta_j$.
If $\beta_q$ were optimal for arbitrarily large times,
Lemma~\ref{lem:optlimit} would imply
$
\Gamma_j^q(t)\to g_q,
$
but $g_q>1$ by maximality of $r$,
which contradicts
$\Gamma_j^q\le1$.
Hence no admissible bid above $\beta_r$
is eventually optimal.
\end{itemize}

\end{proof}

\subsection{Main convergence theorem}

\begin{lemma}
\label{lem:crowded}
$x_1(t)\to\gamma_1^\star$ and $\dot x_1(t)\to0$.
\end{lemma}

\begin{proof}
Ad absurdum, it is clear that the selector for type $1$ shall converge to a deterministic policy. 
The only possible value is the one given by $\gamma_1^\star$, hence $x_1(t)\to\gamma_1^\star$. Hence at some point, $\beta_{l_1}$ is the only element in $\BR_1$. The differential inclusion implies that  $\dot x_1(t)\to0$.
\end{proof}

\begin{lemma}
\label{lem:induction-step}
If $H(j)$ holds, then $x_j(t)\to\gamma_j^\star$ and $\dot x_j(t)\to0$.
\end{lemma}

\begin{proof}
If $\pi_j^\star=0$, apply Lemma~\ref{lem:crowded}. Suppose $\pi_j^\star>0$. Lemma~\ref{lem:rconv} identifies the eventual active block. For $\ell\le s\le r$, Lemma~\ref{lem:optlimit} gives $\Gamma_j^s\to\Gamma_j^{\star,s}$. Below the block, $\Gamma_j^s\le\Gamma_j^\ell\to0$; above it, $\Gamma_j^s\ge\Gamma_j^{r+1}\to1$. Hence $x_j(t)\to\gamma_j^\star$.

For $\ell\le s\le r$, Lemma~\ref{lem:coincide} and \eqref{eq:selector-cum} give
$\dot\Gamma_j^s=-c_s\Gamma_j^\ell+\rho^{-1}(c_s\dot\Psi_j^\ell-\dot\Psi_j^s)\to0$.
Outside the block, both the selector and the state tails converge to $0$ or $1$, so \eqref{eq:gamma-dynamics} again gives $\dot\Gamma_j^s\to0$. Therefore $\dot x_j(t)\to0$.
\end{proof}

\begin{theorem}
\label{thm:main}
Every solution of zero-on-tie fictitious play converges to  the Nash equilibrium $\gamma^\star$. 
\end{theorem}

\begin{proof}
Apply Lemma~\ref{lem:induction-step} successively for $j=1,\dots,m$. The hypothesis $H(1)$ is Lemma~\ref{lem:crowded}, and each step establishes $H(j+1)$. Thus $x(t)\to\gamma^\star$, which is a Nash equilibrium by Lemma~\ref{lem:gamma-nash}.
\end{proof}

\begin{theorem}
\label{thm:epsilon-uniform}
Let $\gamma^\star$ be the zero-on-tie equilibrium and let $q^\star(b)=\sum_{k=1}^m\rho_k\gamma_k^\star(b)$ be the induced  bid mass. Then $\gamma^\star$ is an $\varepsilon$-equilibrium of the uniform-split auction for
$\varepsilon=\tfrac12\max_{j,b}(\theta_j-b)_+q^\star(b)$.
In particular, if $h=\max_bq^\star(b)$, then $\gamma^\star$ is a $\theta_mh/2$-equilibrium.
\end{theorem}

\begin{proof}
Write $\payoff_j^0$ and $\payoff_j^U$ for the zero-on-tie and uniform-split payoffs. Against $\gamma^\star$,
$
  \payoff_j^U(b,\gamma^\star)
  =\payoff_j^0(b,\gamma^\star)
    +\frac12(\theta_j-b)q^\star(b) 
  \le\payoff_j^0(b,\gamma^\star)+\varepsilon.
$
Because $\gamma^\star$ is a zero-on-tie Nash equilibrium,
$\payoff_j^0(b,\gamma^\star)\le\payoff_j^0(\gamma_j^\star,\gamma^\star)$.
The bids in the support of $\gamma_j^\star$ are below value, so the additional tie payoff is nonnegative and
$\payoff_j^0(\gamma_j^\star,\gamma^\star)\le\payoff_j^U(\gamma_j^\star,\gamma^\star)$.
Combining these inequalities proves the deviation bound for pure bids, and linearity extends it to mixed deviations. Finally, $q^\star(b)\le h$ and $(\theta_j-b)_+\le\theta_m$, so $\varepsilon\le\theta_mh/2$.
\end{proof}


\begin{thebibliography}{30}
\expandafter\ifx\csname natexlab\endcsname\relax\def\natexlab#1{#1}\fi
\providecommand{\url}[1]{\texttt{#1}}
\providecommand{\href}[2]{#2}
\providecommand{\path}[1]{#1}
\providecommand{\DOIprefix}{doi:}
\providecommand{\ArXivprefix}{arXiv:}
\providecommand{\URLprefix}{URL: }
\providecommand{\Pubmedprefix}{pmid:}
\providecommand{\doi}[1]{\href{http://dx.doi.org/#1}{\path{#1}}}
\providecommand{\Pubmed}[1]{\href{pmid:#1}{\path{#1}}}
\providecommand{\bibinfo}[2]{#2}
\ifx\xfnm\relax \def\xfnm[#1]{\unskip,\space#1}\fi
\bibitem[{Ahunbay and Bichler(2025)}]{ahunbay2025uniqueness}
\bibinfo{author}{Ahunbay, M.{\c{S}}.}, \bibinfo{author}{Bichler, M.},
  \bibinfo{year}{2025}.
\newblock \bibinfo{title}{On the uniqueness of bayesian coarse correlated
  equilibria in standard first-price and all-pay auctions}, in:
  \bibinfo{booktitle}{Proceedings of the 2025 Annual ACM-SIAM Symposium on
  Discrete Algorithms (SODA)}, \bibinfo{organization}{SIAM}. pp.
  \bibinfo{pages}{2491--2537}.
\bibitem[{Aubin and Cellina(1984)}]{aubinDifferentialInclusionsSetValued1984}
\bibinfo{author}{Aubin, J.P.}, \bibinfo{author}{Cellina, A.},
  \bibinfo{year}{1984}.
\newblock \bibinfo{title}{Differential Inclusions: Set-Valued Maps and
  Viability Theory}. volume \bibinfo{volume}{264} of
  \textit{\bibinfo{series}{Grundlehren der mathematischen Wissenschaften}}.
\newblock \bibinfo{publisher}{Springer}.
\bibitem[{Bena{\"i}m et~al.(2005)Bena{\"i}m, Hofbauer and
  Sorin}]{benaim2005stochastic}
\bibinfo{author}{Bena{\"i}m, M.}, \bibinfo{author}{Hofbauer, J.},
  \bibinfo{author}{Sorin, S.}, \bibinfo{year}{2005}.
\newblock \bibinfo{title}{Stochastic approximations and differential
  inclusions}.
\newblock \bibinfo{journal}{SIAM Journal on Control and Optimization}
  \bibinfo{volume}{44}, \bibinfo{pages}{328--348}.
\bibitem[{Berger(2005)}]{berger2005fictitious}
\bibinfo{author}{Berger, U.}, \bibinfo{year}{2005}.
\newblock \bibinfo{title}{Fictitious play in $2\times n$ games}.
\newblock \bibinfo{journal}{Journal of Economic Theory} \bibinfo{volume}{120},
  \bibinfo{pages}{139--154}.
\bibitem[{Bichler et~al.(2021)Bichler, Fichtl, Heidekr{\"u}ger, Kohring and
  Sutterer}]{bichler2021learning}
\bibinfo{author}{Bichler, M.}, \bibinfo{author}{Fichtl, M.},
  \bibinfo{author}{Heidekr{\"u}ger, S.}, \bibinfo{author}{Kohring, N.},
  \bibinfo{author}{Sutterer, P.}, \bibinfo{year}{2021}.
\newblock \bibinfo{title}{Learning equilibria in symmetric auction games using
  artificial neural networks}.
\newblock \bibinfo{journal}{Nature machine intelligence} \bibinfo{volume}{3},
  \bibinfo{pages}{687--695}.
\bibitem[{Bichler et~al.(2025a)Bichler, Fichtl and
  Oberlechner}]{bichler2025computing}
\bibinfo{author}{Bichler, M.}, \bibinfo{author}{Fichtl, M.},
  \bibinfo{author}{Oberlechner, M.}, \bibinfo{year}{2025}a.
\newblock \bibinfo{title}{Computing bayes--nash equilibrium strategies in
  auction games via simultaneous online dual averaging}.
\newblock \bibinfo{journal}{Operations Research} \bibinfo{volume}{73},
  \bibinfo{pages}{1102--1127}.
\bibitem[{Bichler et~al.(2023)Bichler, Kohring and
  Heidekr{\"u}ger}]{bichler2023learning}
\bibinfo{author}{Bichler, M.}, \bibinfo{author}{Kohring, N.},
  \bibinfo{author}{Heidekr{\"u}ger, S.}, \bibinfo{year}{2023}.
\newblock \bibinfo{title}{Learning equilibria in asymmetric auction games}.
\newblock \bibinfo{journal}{INFORMS Journal on Computing} \bibinfo{volume}{35},
  \bibinfo{pages}{523--542}.
\bibitem[{Bichler et~al.(2025b)Bichler, Lunowa, Oberlechner, Pieroth and
  Wohlmuth}]{bichler2025beyond}
\bibinfo{author}{Bichler, M.}, \bibinfo{author}{Lunowa, S.B.},
  \bibinfo{author}{Oberlechner, M.}, \bibinfo{author}{Pieroth, F.R.},
  \bibinfo{author}{Wohlmuth, B.}, \bibinfo{year}{2025}b.
\newblock \bibinfo{title}{Beyond monotonicity: On the convergence of learning
  algorithms in standard auction games}, in: \bibinfo{booktitle}{Proceedings of
  the AAAI Conference on Artificial Intelligence}, pp.
  \bibinfo{pages}{13649--13657}.
\bibitem[{Brown(1951)}]{brown1951iterative}
\bibinfo{author}{Brown, G.W.}, \bibinfo{year}{1951}.
\newblock \bibinfo{title}{Iterative solutions of games by fictitious play}, in:
  \bibinfo{editor}{Koopmans, T.C.} (Ed.), \bibinfo{booktitle}{Activity Analysis
  of Production and Allocation}. \bibinfo{publisher}{Wiley}, pp.
  \bibinfo{pages}{374--376}.
\bibitem[{Fibich and Gavish(2011)}]{fibich2011numerical}
\bibinfo{author}{Fibich, G.}, \bibinfo{author}{Gavish, N.},
  \bibinfo{year}{2011}.
\newblock \bibinfo{title}{Numerical simulations of asymmetric first-price
  auctions}.
\newblock \bibinfo{journal}{Games and Economic Behavior} \bibinfo{volume}{73},
  \bibinfo{pages}{479--495}.
\bibitem[{Fibich and Gavish(2012)}]{fibich2012asymmetric}
\bibinfo{author}{Fibich, G.}, \bibinfo{author}{Gavish, N.},
  \bibinfo{year}{2012}.
\newblock \bibinfo{title}{Asymmetric first-price auctions---a dynamical-systems
  approach}.
\newblock \bibinfo{journal}{Mathematics of Operations Research}
  \bibinfo{volume}{37}, \bibinfo{pages}{219--243}.
\bibitem[{Filos-Ratsikas et~al.(2021)Filos-Ratsikas, Giannakopoulos, Hollender,
  Lazos and Po{\c{c}}as}]{filosratsikas2021complexity}
\bibinfo{author}{Filos-Ratsikas, A.}, \bibinfo{author}{Giannakopoulos, Y.},
  \bibinfo{author}{Hollender, A.}, \bibinfo{author}{Lazos, P.},
  \bibinfo{author}{Po{\c{c}}as, D.}, \bibinfo{year}{2021}.
\newblock \bibinfo{title}{On the complexity of equilibrium computation in
  first-price auctions}, in: \bibinfo{booktitle}{Proceedings of the 22nd ACM
  Conference on Economics and Computation (EC)}.
\bibitem[{Foster and Young(1998)}]{FOSTER199879}
\bibinfo{author}{Foster, D.P.}, \bibinfo{author}{Young, H.},
  \bibinfo{year}{1998}.
\newblock \bibinfo{title}{On the nonconvergence of fictitious play in
  coordination games}.
\newblock \bibinfo{journal}{Games and Economic Behavior} \bibinfo{volume}{25},
  \bibinfo{pages}{79--96}.
\newblock \URLprefix
  \url{https://www.sciencedirect.com/science/article/pii/S0899825697906266},
  \DOIprefix\doi{https://doi.org/10.1006/game.1997.0626}.
\bibitem[{Gaunersdorfer and Hofbauer(1995)}]{gaunersdorfer1995fictitious}
\bibinfo{author}{Gaunersdorfer, A.}, \bibinfo{author}{Hofbauer, J.},
  \bibinfo{year}{1995}.
\newblock \bibinfo{title}{Fictitious play, shapley polygons, and the replicator
  equation}.
\newblock \bibinfo{journal}{Games and Economic Behavior} \bibinfo{volume}{11},
  \bibinfo{pages}{279--303}.
\bibitem[{Heymann(2025)}]{fp4fpa}
\bibinfo{author}{Heymann, B.}, \bibinfo{year}{2025}.
\newblock \bibinfo{title}{{FP4FPA}: Fictitious play for first-price auctions}.
\newblock \bibinfo{note}{Software repository,
  \url{https://github.com/BenHey/FP4FPA}}.
\bibitem[{Heymann and Mertikopoulos(2021)}]{heymann2021heuristic}
\bibinfo{author}{Heymann, B.}, \bibinfo{author}{Mertikopoulos, P.},
  \bibinfo{year}{2021}.
\newblock \bibinfo{title}{A heuristic for estimating nash equilibria in
  first-price auctions with correlated values}.
\newblock \bibinfo{journal}{arXiv preprint arXiv:2108.04506} .
\bibitem[{Heymann and Mertikopoulos(2025)}]{heymann2025empirical}
\bibinfo{author}{Heymann, B.}, \bibinfo{author}{Mertikopoulos, P.},
  \bibinfo{year}{2025}.
\newblock \bibinfo{title}{An empirical study of fictitious play for estimating
  nash equilibria in first-price auctions with correlated values}.
\newblock \bibinfo{note}{DynaFront 2025 Workshop, NeurIPS 2025}.
\bibitem[{Jordan(1993)}]{jordan1993three}
\bibinfo{author}{Jordan, J.S.}, \bibinfo{year}{1993}.
\newblock \bibinfo{title}{Three problems in learning mixed-strategy nash
  equilibria}.
\newblock \bibinfo{journal}{Games and Economic Behavior} \bibinfo{volume}{5},
  \bibinfo{pages}{368--386}.
\bibitem[{Krishna(2009)}]{krishna2009auction}
\bibinfo{author}{Krishna, V.}, \bibinfo{year}{2009}.
\newblock \bibinfo{title}{Auction Theory}.
\newblock \bibinfo{edition}{2nd} ed., \bibinfo{publisher}{Academic Press}.
\bibitem[{Krishna and Sj{\"o}str{\"o}m(1998)}]{krishna1998convergence}
\bibinfo{author}{Krishna, V.}, \bibinfo{author}{Sj{\"o}str{\"o}m, T.},
  \bibinfo{year}{1998}.
\newblock \bibinfo{title}{On the convergence of fictitious play}.
\newblock \bibinfo{journal}{Mathematics of Operations Research}
  \bibinfo{volume}{23}, \bibinfo{pages}{479--511}.
\bibitem[{Marshall et~al.(1994)Marshall, Meurer, Richard and
  Stromquist}]{marshall1994numerical}
\bibinfo{author}{Marshall, R.C.}, \bibinfo{author}{Meurer, M.J.},
  \bibinfo{author}{Richard, J.F.}, \bibinfo{author}{Stromquist, W.},
  \bibinfo{year}{1994}.
\newblock \bibinfo{title}{Numerical analysis of asymmetric first price
  auctions}.
\newblock \bibinfo{journal}{Games and Economic Behavior} \bibinfo{volume}{7},
  \bibinfo{pages}{193--220}.
\bibitem[{Maskin and Riley(2003)}]{maskin2003uniqueness}
\bibinfo{author}{Maskin, E.}, \bibinfo{author}{Riley, J.},
  \bibinfo{year}{2003}.
\newblock \bibinfo{title}{Uniqueness of equilibrium in sealed high-bid
  auctions}.
\newblock \bibinfo{journal}{Games and Economic Behavior} \bibinfo{volume}{45},
  \bibinfo{pages}{395--409}.
\bibitem[{Milgrom and Weber(1982)}]{milgrom1982theory}
\bibinfo{author}{Milgrom, P.R.}, \bibinfo{author}{Weber, R.J.},
  \bibinfo{year}{1982}.
\newblock \bibinfo{title}{A theory of auctions and competitive bidding}.
\newblock \bibinfo{journal}{Econometrica} \bibinfo{volume}{50},
  \bibinfo{pages}{1089--1122}.
\bibitem[{Monderer and Sela(1996)}]{monderer1996a2}
\bibinfo{author}{Monderer, D.}, \bibinfo{author}{Sela, A.},
  \bibinfo{year}{1996}.
\newblock \bibinfo{title}{A $2\times2$ game without the fictitious play
  property}.
\newblock \bibinfo{journal}{Games and Economic Behavior} \bibinfo{volume}{14},
  \bibinfo{pages}{144--148}.
\bibitem[{Monderer and Shapley(1996)}]{monderer1996fictitious}
\bibinfo{author}{Monderer, D.}, \bibinfo{author}{Shapley, L.S.},
  \bibinfo{year}{1996}.
\newblock \bibinfo{title}{Fictitious play property for games with identical
  interests}.
\newblock \bibinfo{journal}{Journal of Economic Theory} \bibinfo{volume}{68},
  \bibinfo{pages}{258--265}.
\bibitem[{Paes~Leme et~al.(2020)Paes~Leme, Sivan and Teng}]{paesleme2020why}
\bibinfo{author}{Paes~Leme, R.}, \bibinfo{author}{Sivan, B.},
  \bibinfo{author}{Teng, Y.}, \bibinfo{year}{2020}.
\newblock \bibinfo{title}{Why do competitive markets converge to first-price
  auctions?}, in: \bibinfo{booktitle}{Proceedings of The Web Conference 2020},
  pp. \bibinfo{pages}{596--605}.
\bibitem[{Robinson(1951)}]{robinson1951iterative}
\bibinfo{author}{Robinson, J.}, \bibinfo{year}{1951}.
\newblock \bibinfo{title}{An iterative method for solving a game}.
\newblock \bibinfo{journal}{Annals of Mathematics} \bibinfo{volume}{54},
  \bibinfo{pages}{296--301}.
\bibitem[{Shapley(1964)}]{shapley1964some}
\bibinfo{author}{Shapley, L.S.}, \bibinfo{year}{1964}.
\newblock \bibinfo{title}{Some topics in two-person games}, in:
  \bibinfo{booktitle}{Advances in Game Theory}. \bibinfo{publisher}{Princeton
  University Press}. number~\bibinfo{number}{52} in \bibinfo{series}{Annals of
  Mathematics Studies}.
\bibitem[{Vickrey(1961)}]{vickrey1961counterspeculation}
\bibinfo{author}{Vickrey, W.}, \bibinfo{year}{1961}.
\newblock \bibinfo{title}{Counterspeculation, auctions, and competitive sealed
  tenders}.
\newblock \bibinfo{journal}{The Journal of Finance} \bibinfo{volume}{16},
  \bibinfo{pages}{8--37}.
\bibitem[{Wang et~al.(2020)Wang, Shen and Zuo}]{wang2020bayesian}
\bibinfo{author}{Wang, Z.}, \bibinfo{author}{Shen, W.}, \bibinfo{author}{Zuo,
  S.}, \bibinfo{year}{2020}.
\newblock \bibinfo{title}{Bayesian nash equilibrium in first-price auction with
  discrete value distributions}, in: \bibinfo{booktitle}{Proceedings of the
  19th International Conference on Autonomous Agents and MultiAgent Systems
  (AAMAS)}, pp. \bibinfo{pages}{1458--1466}.

\end{thebibliography}
\end{document}